\documentclass{jpsj3}

\usepackage{amsmath}
\usepackage{amssymb}
\usepackage{ascmac}
\usepackage{braket}
\usepackage{delarray}
\usepackage{here}

\usepackage{amsthm}
\theoremstyle{plain}
\newtheorem{theorem}{Theorem}[section]
\newtheorem*{theorem*}{Theorem}

\newtheorem*{definition*}{Definition}

\newtheorem*{lemma*}{Lemma}
\newtheorem{corollary}[theorem]{Corollary}

\usepackage{txfonts}

\usepackage{latexsym}

\usepackage{epsfig}
\usepackage{bm}
\usepackage{color}
\usepackage{ulem}

\newcommand{\be}{\begin{eqnarray}}
\newcommand{\ee}{\end{eqnarray}}
\newcommand{\ba}{\begin{array}}
\newcommand{\ea}{\end{array}}
\newcommand{\bmat}{\left(\begin{array}}
\newcommand{\emat}{\end{array}\right)}
\newcommand{\no}{\nonumber}

\title{Some inequalities for correlation functions of Ising models with quenched randomness}

\author{ Manaka Okuyama$^1$\thanks{manaka.okuyama.d2@tohoku.ac.jp} and Masayuki Ohzeki$^{1,2,3}$}
\inst{$^1$Graduate School of Information Sciences, Tohoku University, Sendai 980-8579, Japan \\
$^2$Institute of Innovative Research, Tokyo Institute of Technology, Yokohama 226-8503, Japan \\
$^3$Sigma-i Co. Ltd., Tokyo 108-0075, Japan} 

\abst{
Correlation inequalities have played an essential role in the analysis of ferromagnetic models but have not been established in spin glass models.
In this study, we obtain some correlation inequalities for the Ising models with quenched randomness, where the distribution of the interactions is symmetric.
The acquired inequalities can be regarded as an extension of the previous results, which were limited to the local energy for a spin set, to the local energy for a pair of spin sets.
Besides, we also obtain some correlation inequalities for asymmetric distribution.
}

\begin{document}
\maketitle

\section{Introduction}

Spin-glass models describe spatially disturbed magnetic material.
While the mean-field theory of spin-glass models was established by the full replica symmetry breaking solution~\cite{Parisi,Parisi2,Guerra,Talagrand}, it is a very difficult problem to understand the property of finite-dimensional systems, except on the Nishimori-line~\cite{Nishimori}.

As a mathematical tool, correlation inequalities such as the Griffiths inequalities and the Fortuin-Kasteleyn-Ginibre inequality have played an essential role in the analysis of ferromagnetic models.
Correlation inequalities also make a critical contribution to the Ising model in a random field~\cite{KTZ}, and thus it is naturally expected that correlation inequalities may have an important role in the analysis of finite-dimensional spin-glass models.

There are some previous studies on correlation inequalities in spin-glass models.
Recent studies~\cite{CG,CL} showed that, when the probability distribution function of random interactions is symmetric,  the counterpart of the Griffiths first inequality holds in the Ising models with quenched randomness.
In addition, it was proved that the counterpart of the Griffiths second inequality holds on the Nishimori-line~\cite{MNC, Kitatani} for various bond randomness which includes the Gaussian distribution and the binary distribution.
On the other hand, a current work~\cite{KNA} found the lower bound on the local energy of the Ising with quenched randomness.
As a consequence, for the Gaussian distribution, it was  shown that the expected value of the square of the correlation function always has a finite lower bound at any temperature~\cite{OO}.

 However, correlation inequalities as in ferromagnetic spin models have not been generally established, and rigorous analysis based on correlation inequalities for spin-glass models has been limited to a few examples~\cite{CS,CG,CG2}.
The purpose of this study is to explore the possibility of correlation inequalities for spin-glass models.
Although the previous studies~\cite{CL,KNA} have been limited to the local energy for a spin set, we extend their results to the local energy for a pair of spin sets.
The key ingredient of the proof is a simple representation of correlation functions.
The organization of the paper is as follows.
In Sec. II, we define the model and gives the simple representation of the correlation functions.
In Sec. III, we obtain some correlation inequalities for symmetric distribution, which is a natural extension of the previous studies~\cite{CL,KNA}.
Section IV is devoted to the case of asymmetric distribution.
Finally, our conclusion is given in Sec. V.

\section{Ising model with quenched randomness and simple expression of correlation functions}

Following Ref. \cite{KNA}, we consider a generic form of the Ising model,
\be
H&=&- \sum_{A \subset{V}}   J_{A} \sigma_A , \label{def-model}
\\
\sigma_A &\equiv& \prod_{i\in A} \sigma_i ,
\ee
where $V$ is the set of sites, the sum over $A$ is over all the subsets of $V$ in which interactions exist, and the lattice structure adopts any form.
The probability distribution of a random interaction $J_A$ is represented as $P_A (J_A)$.
The probability distributions can be generally different from each other, i.e., $P_A (x)\neq P_{B} (x)$, and are also allowed to present no randomness, i.e., $P_A(J_A)=\delta(J-J_A)$.

The partition function $Z_{\{J_{A}\}}$ and correlation function $\langle\sigma_B \rangle_{\{J_{A}\}}$ for a set of fixed interactions $\{J_{A}\}$ are given by
\be
Z_{\{J_{A}\}}&=&\Tr \exp\left( \beta \sum_{A \subset{V}}   J_{A} \sigma_A \right)
\\
\langle\sigma_B \rangle_{\{J_{A}\}}&=& \frac{ \Tr \sigma_B \exp\left(\beta \sum_{A \subset{V}}  J_{A} \sigma_A \right)}{Z_{\{J_{A}\}}} .
\ee
The configurational average over the distribution of randomness of the interactions is written as 
\be
\mathbb{E}\left[g(\{J_{A}\})  \right] =  \left(\prod_{A \subset{V}} \int_{-\infty}^\infty dJ_A P_A(J_A) \right) g(\{J_{A}\}).
\ee
For example, the expected value of the correlation function is obtained as
\be
\mathbb{E}\left[ \langle\sigma_B \rangle_{\{J_{A}\}} \right]&=&  \left(\prod_{A \subset{V}} \int_{-\infty}^\infty dJ_A P_A(J_A) \right)   \frac{ \Tr \sigma_B \exp\left(\beta \sum_{A \subset{V}}  J_{A} \sigma_A \right)}{Z_{\{J_{A}\}}} .
\ee

When, we focus on $z_B=\exp(\beta J_B)$, we can generally express $\langle\sigma_B \rangle_{\{J_{A}\}}$ as
\be
\langle\sigma_B \rangle_{\{J_{A}\}}&=&\frac{ x z_B - y z_B^{-1}}{ x z_B + y z_B^{-1}},
\ee
where $x$ and $y$ are positive and do not contain $z_B$.
Similarly, when we are interested in two variables $z_B=\exp(\beta J_B)$ and $z_C=\exp(\beta J_C)$, we can generally represent $\langle\sigma_B \rangle_{\{J_{A}\}}$, $\langle\sigma_C \rangle_{\{J_{A}\}}$ and $\langle\sigma_B \sigma_C \rangle_{\{J_{A}\}}$ as 
\be
\langle  \sigma_B   \rangle_{\{J_{A}\}}&=&  \frac{a z_B z_C +b z_B z_C^{-1} -c z_B^{-1} z_C-d z_B^{-1} z_C^{-1} }{a z_B z_C +b z_B z_C^{-1} +c z_B^{-1} z_C+d z_B^{-1} z_C^{-1}  } ,
\\
\langle  \sigma_C   \rangle_{\{J_{A}\}}&=&  \frac{a z_B z_C -b z_B z_C^{-1} +c z_B^{-1} z_C-d z_B^{-1} z_C^{-1} }{a z_B z_C +b z_B z_C^{-1} +c z_B^{-1} z_C+d z_B^{-1} z_C^{-1}  } ,
\\
\langle  \sigma_B  \sigma_C \rangle_{\{J_{A}\}}&=& \frac{a z_B z_C -b z_B z_C^{-1} -c z_B^{-1} z_C+d z_B^{-1} z_C^{-1} }{a z_B z_C +b z_B z_C^{-1} +c z_B^{-1} z_C+d z_B^{-1} z_C^{-1}  } ,
\ee
where $a, b, c$ and $d$ are positive and don not contain $z_B$ and $z_C$.
These simple expression is very useful in the following calculation.

\section{Correlation inequalities for symmetric distribution}
In this section, we focus on the case that the distribution functions of $J_B$ and $J_C$ are symmetric, 
\be
P_B(-J_B)&=&P_B(J_B),
\\
P_C(-J_C)&=&P_C(J_C).
\ee
We emphasize that we do not impose any constraint on all the other interactions than $J_B$ and $J_C$.

\subsection{Rederivation of Griffiths first inequality for Ising model with quenched randomness}
First, as a exercise, we reproduce the Griffiths first inequality for Ising model with quenched randomness~\cite{CL},
\be
\mathbb{E}\left[ J_B \langle\sigma_B \rangle_{\{J_{A}\}} \right] \ge0 \label{griffiths-1st}.
\ee
\begin{proof}
{\rm
By dividing the integration interval of $J_B$ and summing up them, a simple calculation gives
\be
\mathbb{E} \left[ J_B \langle  \sigma_B \rangle_{\{J_{A}\}} \right]
&=& \int_0^{\infty} dJ_B P_B(J_B)  J_B \mathbb{E} \left[  \frac{2xy (z_B^4-1) }{(x z_B^2 +y)(x +y z_B^2)} \right]' 
\no\\
&\ge&0, \label{simple-cal-1}
\ee
where $\mathbb{E}[\cdots]'$ denotes the configurational average over the randomness of the interactions other than $J_B$.
Thus, we obtain Eq. (\ref{griffiths-1st}).
}
\end{proof}
\subsection{Rederivation of inequality for local energy for a spin set of Ising model with quenched randomness}
Next, we reproduce the inequality for the local energy for a spin set~\cite{KNA},
\be
&& \mathbb{E} \left[ J_B \tanh(\beta J_B)\right] \ge\mathbb{E} \left[ J_B \langle  \sigma_B \rangle_{\{J_{A}\}} \right] \label{kitatani}.
\ee
\begin{proof}
{\rm
Similarly to the derivation of Eq. (\ref{simple-cal-1}), the simple calculation gives
\be
\mathbb{E} \left[ J_B \tanh(\beta J_B)-J_B \langle  \sigma_B \rangle_{\{J_{A}\}} \right]
&=& \int_0^{\infty} dJ_B P_B(J_B)  J_B \mathbb{E} \left[  \frac{2(x-y)^2 z_B^2 (z_B^2-1) }{(1+z_B^2)(x z_B^2 +y)(x +y z_B^2)} \right]' 
\no\\
&\ge&0.
\ee
Thus, we reach Eq. (\ref{kitatani}).
}
\end{proof}

\subsection{Inequalities for local energy for a pair of spin sets }
We have derived the inequalities for the local energy for a spin set so far.
Our first main result in this study is the following theorem.
\begin{theorem}\label{th1}
The quenched average of the local energy for a pair of spin sets, which is defined as $\mathbb{E}\left[ J_B J_C  \langle\sigma_B \sigma_C \rangle_{\{J_{A}\}} \right]$, is always positive, 
\be
\mathbb{E}\left[ J_B J_C  \langle\sigma_B \sigma_C \rangle_{\{J_{A}\}} \right]\ge0 \label{Griffiths-ex} .
\ee
\end{theorem} 
\begin{proof}
{\rm
By dividing the integration interval of $J_B$ and $J_C$ and summing up them, a straightforward calculation provides
\be
&&   \mathbb{E} \left[  J_B J_C\langle  \sigma_B  \sigma_C \rangle_{\{J_{A}\}}\right]
\no\\
&=&\int_0^{\infty} dJ_B P_B(J_B) \int_0^{\infty} dJ_C  P_C(J_C) J_B J_C  2(z_B^4-1)(z_C^4-1) 
\no\\
&&\mathbb{E} \left[  \frac{bcd((b+2d z_B^2+b z_B^4) z_C^2+ c z_B^2(1+z_C^4) )+a^2(cd (1+z_B^4)z_C^2 +b z_B^2(d+2c z_C^2+dz_C^4))}
{(a z_B^2 z_C^2 +b z_B^2  +c  z_C^2+d )(c z_B^2 z_C^2 +d z_B^2  +a  z_C^2+b)(b z_B^2 z_C^2 +a z_B^2  +d  z_C^2+c )(d z_B^2 z_C^2 +c z_B^2  +b  z_C^2+a) } \right.
\no\\
&&+ \left.  \frac{a(b^2z_B^2(c+2dz_C^2+cz_C^4) +cd z_B^2(d+2c z_C^2+dz_C^4)+b(1+z_B^4)(c^2z_C^2+d^2z_C^2+2cd(1+z_C^4)))}{(a z_B^2 z_C^2 +b z_B^2  +c  z_C^2+d )(c z_B^2 z_C^2 +d z_B^2  +a  z_C^2+b)(b z_B^2 z_C^2 +a z_B^2  +d  z_C^2+c )(d z_B^2 z_C^2 +c z_B^2  +b  z_C^2+a)  } \right]'' ,\label{simple-cal-2}
\no\\
\ee
where $\mathbb{E}[\cdots]''$ denotes the configurational average over the randomness of the interactions other than $J_B$ and $J_C$.
Therefore, we obtain Eq. (\ref{Griffiths-ex}).
}
\end{proof}

  Inequality (\ref{Griffiths-ex}) can be regarded as a natural extension of the Griffiths first inequality (\ref{griffiths-1st}) to the local energy for a pair of spin sets.
We note that, in the case of the Gaussian distribution with the variance $\Lambda_B^2$ , Eq. (\ref{Griffiths-ex}) can be proved directly by integration by parts, 
\be
&&\mathbb{E}\left[ J_B J_C  \langle\sigma_B \sigma_C \rangle_{\{J_{A}\}} \right]
\no\\
&=&\Lambda_B^2 \Lambda_C^2 \mathbb{E} \left[ 1- \langle   \sigma_C \rangle_{\{J_{A}\}}^2 - \langle  \sigma_B  \rangle_{\{J_{A}\}}^2 -\langle  \sigma_B  \sigma_C \rangle_{\{J_{A}\}}^2+ 2\langle  \sigma_B  \sigma_C \rangle_{\{J_{A}\}} \langle  \sigma_B  \rangle_{\{J_{A}\}} \langle  \sigma_C  \rangle_{\{J_{A}\}} \right]
\no\\
&\ge&0,
\ee
where we used the following inequality 
\be
1- \langle   \sigma_C \rangle_{\{J_{A}\}}^2 - \langle  \sigma_B  \rangle_{\{J_{A}\}}^2 -\langle  \sigma_B  \sigma_C \rangle_{\{J_{A}\}}^2+ 2\langle  \sigma_B  \sigma_C \rangle_{\{J_{A}\}} \langle  \sigma_B  \rangle_{\{J_{A}\}} \langle  \sigma_C  \rangle_{\{J_{A}\}}\ge0 ,\label{general-ineq}
\ee
which generally holds in the Ising models (see Appendix for details of the proof).


Our second main result is the following inequality which gives the opposite bound for the quenched average of the local energy for a pair of spin sets.
\begin{theorem}\label{th1}
For $\sigma_B\neq \sigma_C$, the quenched average of the local energy for a pair of spin sets is bounded from above,
\be
\mathbb{E} \left[ J_B J_C \tanh(\beta J_B)\tanh(\beta J_C) \right] \ge \mathbb{E} \left[ J_B J_C \langle  \sigma_B  \sigma_C \rangle_{\{J_{A}\}} \right].\label{two-kitatani}
\ee
\end{theorem}
This inequality has a clear physical meaning: the quenched average of the local energy for a pair of spin sets is always lower than or equal to the energy in the absence of all the other interactions.
Equation (\ref{two-kitatani}) is a natural extension of Eq. (\ref{kitatani}) to a pair of spin sets.
\begin{proof}
{\rm
Similarly to the derivation of Eq. (\ref{simple-cal-2}), the straightforward but tedious calculation gives
\be
&&\mathbb{E} \left[ J_B J_C \tanh(\beta J_B)\tanh(\beta J_C)- J_B J_C \langle  \sigma_B  \sigma_C \rangle_{\{J_{A}\}} \right]
\no\\
&=& \int_0^{\infty} dJ_B P_B(J_B) \int_0^{\infty} dJ_C  P_C(J_C)  2J_B J_C  (z_B^2-1)(z_C^2-1) \mathbb{E} \left[  \frac{B}
{A} \right]'' 
\ee
where $A$ and $B$ are defined as
\be
A&=&(1+z_B^2)(1+z_C^2)(d+bz_B^2+c z_C^2 +a z_B^2 z_C^2) 
\no\\
&&(c+a z_B^2+d z_C^2 +b z_B^2 z_C^2) (b+d z_B^2+a z_C^2 +c z_B^2 z_C^2) (a+c z_B^2+b z_C^2 +d z_B^2 z_C^2) ,
\\
B&=&2(a^2b^2- ab^2c +b^2c^2-a^2bd-bc^2d+a^2d^2-acd^2+c^2d^2)(z_B^4+z_B^4 z_C^8) 
\no\\
&&+2(a^2c^2- abc^2 +b^2c^2-a^2cd-b^2cd+a^2d^2-abd^2+b^2d^2)(z_C^4+z_C^4 z_B^8)
\no\\
&&+(2a^4+8a^2b^2+2b^4-8a^2bc-4ab^2c+8a^2c^2-4abc^2+4b^2c^2+2c^4-4a^2bd-8ab^2d-4a^2cd
\no\\
&&+16abcd-4b^2cd-8ac^2d-4bc^2d+4a^2d^2-4abd^2+8b^2d^2-4acd^2-8bcd^2+8c^2d^2+2d^4)z_B^4z_C^4
\no\\
&&+(ab^2c+a^2bd-4abcd+bc^2d+acd^2)(z_B^2+z_B^6+z_B^2z_C^8+z_B^6z_C^8)
\no\\
&&+(abc^2+a^2cd-4abcd+b^2cd+abd^2)(z_C^2+z_C^6+z_B^8z_C^2+z_B^8z_C^6)
\no\\
&&+(4a^3b+4ab^3-4a^2bc-4ab^2c+6abc^2-4a^2bd-4ab^2d+6a^2cd-8abcd+6b^2cd
\no\\
&&-4ac^2d-4bc^2d+4c^3d+6abd^2-4acd^2-4bcd^2+4cd^3)(z_B^4z_C^2+z_B^4z_C^6)
\no\\
&&+(4a^3c-4a^2bc+6ab^2c-4abc^2+4ac^3+6a^2bd-4ab^2d+4b^3d-4a^2cd-8abcd
\no\\
&&-4b^2cd-4ac^2d+6bc^2d-4abd^2+6acd^2-4bcd^2+4bd^3)(z_B^2z_C^4+z_B^6z_C^4)
\no\\
&&+(4a^2bc-2ab^2c+2b^3c-2abc^2+2bc^3+2a^3d-2a^2bd+4ab^2d-2a^2cd-8abcd
\no\\
&&-2b^2cd+4ac^2d-2bc^2d-abd^2-2acd^2+4bcd^2+2ad^3)(z_B^2z_C^2+z_B^6z_C^2+z_B^2z_C^6+z_B^6z_C^6).
\ee
$A$ is obviously positive.
Furthermore, all of the terms in B are positive, because
\be
&&2(a^2b^2- ab^2c +b^2c^2-a^2bd-bc^2d+a^2d^2-acd^2+c^2d^2)\ge0,
\no\\
&&2(a^2c^2- abc^2 +b^2c^2-a^2cd-b^2cd+a^2d^2-abd^2+b^2d^2)\ge0,
\no\\
&&(2a^4+8a^2b^2+2b^4-8a^2bc-4ab^2c+8a^2c^2-4abc^2+4b^2c^2+2c^4-4a^2bd-8ab^2d-4a^2cd
\no\\
&&+16abcd-4b^2cd-8ac^2d-4bc^2d+4a^2d^2-4abd^2+8b^2d^2-4acd^2-8bcd^2+8c^2d^2+2d^4)\ge0,
\no\\
&&(ab^2c+a^2bd-4abcd+bc^2d+acd^2)\ge0,
\no\\
&&(abc^2+a^2cd-4abcd+b^2cd+abd^2)\ge0,
\no\\
&&(4a^3b+4ab^3-4a^2bc-4ab^2c+6abc^2-4a^2bd-4ab^2d+6a^2cd-8abcd+6b^2cd
\no\\
&&-4ac^2d-4bc^2d+4c^3d+6abd^2-4acd^2-4bcd^2+4cd^3)\ge0,
\no\\
&&(4a^3c-4a^2bc+6ab^2c-4abc^2+4ac^3+6a^2bd-4ab^2d+4b^3d-4a^2cd-8abcd
\no\\
&&-4b^2cd-4ac^2d+6bc^2d-4abd^2+6acd^2-4bcd^2+4bd^3)\ge0,
\no\\
&&(4a^2bc-2ab^2c+2b^3c-2abc^2+2bc^3+2a^3d-2a^2bd+4ab^2d-2a^2cd-8abcd
\no\\
&&-2b^2cd+4ac^2d-2bc^2d-abd^2-2acd^2+4bcd^2+2ad^3)\ge0.
\ee
Therefore, we obtain Eq. (\ref{two-kitatani}).
}
\end{proof}

Furthermore, when we consider the Gaussian distribution as the distribution of the interactions, by integration by parts, Eq. (\ref{two-kitatani}) can be rewritten as follows.
\begin{corollary} 
We consider the case where the interactions $J_B$ and $J_C$ follow the Gaussian distributions with the variance $\Lambda_B^2$ and $\Lambda_C^2$, respectively.
Then, for $\sigma_B\neq \sigma_C$, the following relation holds,
\be
 \mathbb{E} \left[ \langle  \sigma_B  ;  \sigma_C  \rangle_{\{J_{A}\}}^2 \right]\ge \mathbb{E} \left[(1- \langle  \sigma_B  \rangle_{\{J_{A}\}}^2 )(1- \langle  \sigma_C  \rangle_{\{J_{A}\}}^2 )-\left(1-\tanh^2(\beta J_B)  \right)\left(1-\tanh^2(\beta J_C)  \right) \right] .
\ee
\end{corollary}

\subsection{Another generalization of Griffiths first inequality}
We have considered the quenched average of the local energy for a pair of spin sets.
Using the simple representation of correlation functions, we can find another generalization of the Griffiths first inequality (\ref{griffiths-1st}). 
Our result is as follows.
\begin{theorem}\label{th1}
For  $1\le x_1, x_2 \le1$ and any product of the Ising variables $\sigma_C$, the following relation hods,
\be
\mathbb{E}\left[ J_B \left(\langle\sigma_B \rangle_{\{J_{A}\}} +x_1\langle\sigma_C \rangle_{\{J_{A}\}} +x_2\langle\sigma_B \sigma_C \rangle_{\{J_{A}\}}  \right)\right] \ge0 . \label{general-Griffiths}
\ee
Similarly, for $-1\le x_1,x_2,x_3,x_4 \le1$ and any products of the Ising variables $\sigma_C$ and $\sigma_D$, the following inequality hods,
\be
\mathbb{E}\left[ J_B \left(2\langle\sigma_B \rangle_{\{J_{A}\}} +x_1\langle\sigma_C \rangle_{\{J_{A}\}} +x_2\langle\sigma_D \rangle_{\{J_{A}\}}+x_3\langle\sigma_B \sigma_C \rangle_{\{J_{A}\}}+x_4\langle\sigma_D \sigma_B \rangle_{\{J_{A}\}}  \right)\right] \ge0 .
\ee
Furthermore, for $-1\le x_1,x_2,x_3,x_4, x_5, x_6 \le1$ and any products of the Ising variables $\sigma_C$ and $\sigma_D$, the following relation hods,
\be
&&\mathbb{E}\left[ J_B \left(3\langle\sigma_B \rangle_{\{J_{A}\}} +x_1\langle\sigma_C \rangle_{\{J_{A}\}} +x_2\langle\sigma_D \rangle_{\{J_{A}\}}+x_3\langle\sigma_B \sigma_C \rangle_{\{J_{A}\}}  \right.\right.
\no\\
&&\left. \left.+x_4\langle\sigma_C \sigma_D \rangle_{\{J_{A}\}}+x_5\langle\sigma_D \sigma_B \rangle_{\{J_{A}\}} +x_6\langle\sigma_B \sigma_C\sigma_D  \rangle_{\{J_{A}\}}  \right)\right] 
\ge0 .
\ee
\end{theorem}
\begin{proof}
{\rm
When we only focus on $z_B=\exp(\beta  J_B)$, for any product of the Ising variables $\sigma_C$, we can express $\langle\sigma_B \rangle_{\{J_{A}\}}$, $\langle\sigma_C \rangle_{\{J_{A}\}}$ and $\langle\sigma_B \sigma_C \rangle_{\{J_{A}\}}$ as
\be
\langle  \sigma_B   \rangle_{\{J_{A}\}}&=&  \frac{e z_B  +f z_B -g z_B^{-1}-h z_B^{-1}  }{e z_B  +f z_B  +g z_B^{-1} +h z_B^{-1}   } ,
\\
\langle  \sigma_C   \rangle_{\{J_{A}\}}&=&  \frac{e z_B  -f z_B +g z_B^{-1}-h z_B^{-1}  }{e z_B  +f z_B  +g z_B^{-1} +h z_B^{-1}   }  ,
\\
\langle  \sigma_B  \sigma_C \rangle_{\{J_{A}\}}&=&\frac{e z_B  -f z_B -g z_B^{-1}+h z_B^{-1}  }{e z_B  +f z_B  +g z_B^{-1} +h z_B^{-1}   }  ,
\ee
where $e$, $f$, $g$ and $h$ are positive and don't contain $z_B$.
Then, by the same method as before, the direct calculation shows 
\be
&&\mathbb{E}\left[ J_B \left(\langle\sigma_B \rangle_{\{J_{A}\}} +x_1\langle\sigma_C \rangle_{\{J_{A}\}} +x_2\langle\sigma_B \sigma_C \rangle_{\{J_{A}\}}  \right)\right] 
\no\\
&=&\int_0^{\infty} dJ_B P_B(J_B)  J_B \mathbb{E} \left[  \frac{2 (z_B^4-1)\left\{e((1+x_2)g+(1+x_1)h)+f((1-x_1)g+h(1-x_2)) \right\}}{(g+h+(e+f)z_B^2)((g+h)z_B^2+e+f)} \right]' 
\no\\
&\ge&0.
\ee
Thus, we prove Eq. (\ref{general-Griffiths}).
Other inequalities can be also proved by the same manner.
}
\end{proof}

\subsection{Possible inequalities}
From the above results, it is expected that the following quantities are positive,
\be
&&\mathbb{E}\left[ J_B J_C  \langle\sigma_B\rangle_{\{J_{A}\}} \langle \sigma_C  \rangle_{\{J_{A}\}} \right],
\\
&&\mathbb{E}\left[ J_B J_C J_D  \langle\sigma_B \sigma_C \sigma_D \rangle_{\{J_{A}\}} \right], \label{Griffiths-3body}
\\
&&\mathbb{E} \left[ J_B J_C J_D \tanh(\beta J_B)\tanh(\beta J_C)\tanh(\beta J_D) - J_B J_C J_D \langle  \sigma_B  \sigma_C \sigma_D \rangle_{\{J_{A}\}} \right] . \label{Kitatani-3body}
\ee
However, numerical calculations imply that the above quantities do not have a definite sign.

On the other hand, numerical calculation suggests that the following quantity is positive,
\be
&&\mathbb{E} \left[ J_B J_C \tanh(\beta J_B)\tanh(\beta J_C)- J_B J_C  \langle  \sigma_B\rangle_{\{J_{A}\}} \langle  \sigma_C  \rangle_{\{J_{A}\}} \right] ,
\ee
but we have not found a general proof or counter example.
\section{Correlation inequalities for asymmetric distribution}
In this section, we consider the case that the distribution functions of $J_B$ and $J_C$ have a ferromagnetic bias,
\be
P_B(-J_B)&=&\exp(-2\beta_{NL,B} J_B) P_B(J_B), \label{ferro-bias}
\\
P_C(-J_C)&=&\exp(-2\beta_{NL,C} J_C) P_C(J_C),\label{ferro-bias2}
\ee
where $\beta_{NL,B} $ and $\beta_{NL,C}$ are positive.
For example, in the case of the Gaussian distribution
\be
P_B(J_B)&=&\frac{1}{\sqrt{2\pi J^2}}\exp\left(-\frac{(J_B-J_0)^2}{2J^2} \right),
\ee 
and the binary distribution
\be
P_B(J_B)&=&p\delta(J_B-J)+(1-p)\delta(J_B+J) ,
\ee
$\beta_{NL,B} $ is given as follows, respectively,
\be
\beta_{NL,B} &=& \frac{J_0}{J^2} ,
\\
\beta_{NL,B} &=&\frac{p}{1-p}.
\ee

First, we focus on the quenched average of the correlation function $\mathbb{E} \left[ \langle  \sigma_B \rangle_{\{J_{A}\}} \right] $ .
We find the following result.
\begin{theorem}\label{}
Under the condition (\ref{ferro-bias}), the quenched average of the correlation function satisfies
\be
\mathbb{E} \left[ \langle  \sigma_B \rangle_{\{J_{A}\}} \right] \ge \mathbb{E} \left[ \exp(-2\beta_{B} J_B) \langle  \sigma_B \rangle_{\{J_{A}\}} \right]. \label{bias-ineq}
\ee
\end{theorem}
\begin{proof}
{\rm
A straightforward calculation gives
\be
\mathbb{E} \left[ (1-\exp(-2\beta_{B} J_B)) \langle  \sigma_B \rangle_{\{J_{A}\}} \right]
&=& \int_0^{\infty} dJ_B P_B(J_B) (1-\exp(-2\beta_{B} J_B)) \mathbb{E} \left[  \frac{2xy (z_B^4-1) }{(x z_B^2 +y)(x +y z_B^2)} \right]' 
\no\\
&\ge&0 .
\ee
}
\end{proof}
We note that $P_B(-J_B)=e^{-2\beta_{B} J_B} P_B(J_B)$.
Thus, this inequality means that the correlation function with a ferromagnetic bias is always larger than or equal to the one with an antiferromagnetic bias, independent of any other interaction.
This is a natural consequence and it is easy to prove the above inequality for the Gaussian distribution by differentiation; it is not trivial for general distributions such as the binary distribution.

The essence of the proof of Eq. (\ref{bias-ineq}) is to attribute a biased quantity to the quantity calculated in Sec. III.
Thus, following the same manner, we can easily find similar correlation inequalities.

\begin{theorem}\label{}
Under the conditions (\ref{ferro-bias}) and (\ref{ferro-bias2}), the following inequalities hold,
\be
&&\mathbb{E} \left[ J_B(1+\exp(-2\beta_{B} J_B))\langle  \sigma_B \rangle_{\{J_{A}\}} \right] \ge 0 ,
\\
&&\mathbb{E} \left[ \left(1-\exp(-2\beta_{B} J_B)\right)\left(1-\exp(-2\beta_{C} J_C)  \right)  \langle \sigma_B \sigma_C \rangle_{\{J_{A}\}} \right]\ge0,
\\
&& \mathbb{E} \left[ J_B J_C \left(1+\exp(-2\beta_{B} J_B)\right)\left(1+\exp(-2\beta_{C} J_C)  \right)  \langle \sigma_B \sigma_C \rangle_{\{J_{A}\}} \right]\ge0.
\ee
\end{theorem}
\section{ Conclusions}
We have obtained some correlation inequalities for the Ising models with quenched randomness.
Our main inequalities (\ref{Griffiths-ex}) and (\ref{two-kitatani})  are natural extension of previous studies~\cite{CL,KNA} to the local energy for a pair of spin sets.
On the other hand, numerical calculation implied that similar inequalities do not hold in the local energy for a group of spin sets as in Eq. (\ref{Griffiths-3body}) and (\ref{Kitatani-3body}).

In addition, using the calculations of the symmetric distribution, we attain some correlation inequalities for the asymmetric distribution.

Our proof relied strongly on the simple expression of correlation functions.
For higher order correlations, it is hard to find similar inequalities because the expansion terms increase exponentially.
Thus, in order to search for further correlation inequalities for spin-glass models, it is necessary to invent an efficient and systematic method.

\begin{acknowledgment}
The present work was financially supported by JSPS KAKENHI Grant No. 18H03303, 19H01095, 19K23418, and the JST-CREST (No.JPMJCR1402) for Japan Science and Technology Agency.
\end{acknowledgment}

\appendix
\section{Proof of Eq. (\ref{general-ineq})}\label{appendix}
For the Ising Hamiltonian (\ref{def-model}), arbitrary correlation functions $\langle \sigma_B \sigma_C \rangle,\langle \sigma_B \rangle,$ and $\langle\sigma_C \rangle$ can be generally represented as
\be
\langle   \sigma_B \sigma_C \rangle&=&\frac{a-b-c+d}{a+b+c+d},
\\
\langle   \sigma_B \rangle&=&\frac{a+b-c-d}{a+b+c+d},
\\
\langle   \sigma_C \rangle&=&\frac{a-b+c-d}{a+b+c+d},
\ee
where $a$, $b$, $c$ and $d$ are always positive and depend on the model details, and $\sigma_B$ and $\sigma_C$ are the product of arbitrary Ising variables,
\be
\sigma_B&\equiv& \prod_{i\in B} \sigma_i ,
\\
\sigma_C&\equiv& \prod_{i\in C} \sigma_i .
\ee
We note that $\sigma_B$ and $\sigma_C$ take only $\pm1$, which allows for the expression of the above equations.
Then, a straightforward calculation shows 
\be
&&1- \langle   \sigma_C \rangle^2 - \langle  \sigma_B  \rangle^2 -\langle  \sigma_B  \sigma_C \rangle^2+ 2\langle  \sigma_B  \sigma_C \rangle \langle  \sigma_B  \rangle \langle  \sigma_C \rangle 
\no\\
&=&16\frac{acd+bcd+abc+abd}{(a+b+c+d)^3}
\no\\
&\ge&0.
\ee
Thus, we prove Eq. (\ref{general-ineq})


\end{document}